\documentclass[11pt]{article}

\pdfoutput=1

\usepackage{graphicx}
\usepackage{amssymb, amsthm, amsmath}
\usepackage{subfigure}
\usepackage{fullpage}

\newtheorem{theorem}{Theorem}

\newtheorem{definition}[theorem]{Definition}
\newtheorem{lemma}[theorem]{Lemma}

\newcommand{\paren}[1]{\left({#1}\right)}
\newcommand{\braces}[1]{\left\{{#1}\right\}}
\newcommand{\bracks}[1]{\left[{#1}\right]}

\newcommand{\setnot}[2]{\braces{#1 \mbox{ }: \mbox{ } #2}}

\renewcommand{\phi}{\varphi}

\DeclareMathOperator \pr {\mathbb{P}r}
\DeclareMathOperator \expect {\mathbb{E}}
\DeclareMathOperator \ranking {Ranking}
\DeclareMathOperator \rankingsim {RankingSimulate}

\title{\bf Thinking Twice about Second-Price Ad Auctions}
\author{
  Yossi Azar\footnote{
    Microsoft Research and Tel Aviv University.  \texttt{azar@tau.ac.il}
  } \and
  Benjamin Birnbaum\footnote{
    University of Washington, Department of Computer Science and Engineering.
    Research supported by an NSF Graduate Research Fellowship.
    \texttt{birnbaum@cs.washington.edu}
  } \and
  Anna R. Karlin\footnote{
    University of Washington, Department of Computer Science and Engineering.
    Research supported by NSF Grant CCF-0635147 and a grant from Yahoo! Research.
    \texttt{karlin@cs.washington.edu}
  } \and
  C.~Thach Nguyen\footnote{
    University of Washington, Department of Computer Science and Engineering.
    Research supported in part by NSF Grant CCF-0635147 and a grant from Yahoo! Research.
    \texttt{ncthach@cs.washington.edu}
  }
}
\date{ }

\begin{document}

\maketitle

\begin{abstract}
  A number of recent papers have addressed the algorithmic problem of
  allocating advertisement space for keywords in sponsored search
  auctions so as to maximize revenue, most of which assume that
  pricing is done via a first-price auction.  This does not
  realistically model the Generalized Second Price (GSP) auction used
  in practice, in which bidders pay the next-highest bid for keywords
  that they are allocated.  Towards the goal of more realistically
  modelling these auctions, we introduce the {\em Second-Price Ad
    Auctions} problem, in which bidders' payments are determined by
  the GSP mechanism.

  We show that the complexity of the Second-Price Ad Auctions problem
  is quite different than that of the more studied First-Price Ad
  Auctions problem.  First, unlike the first-price variant, for which
  small constant-factor approximations are known, we show that it is
  NP-hard to approximate the Second-Price Ad Auctions problem to any
  non-trivial factor, even when the bids are small compared to the
  budgets.  Second, we show that this discrepancy extends even to the
  $0$-$1$ special case that we call the {\em Second-Price Matching}
  problem (2PM).  In particular, offline 2PM is APX-hard, and for
  online 2PM there is no deterministic algorithm achieving a
  non-trivial competitive ratio and no randomized algorithm achieving
  a competitive ratio better than $2$.  This stands in contrast to
  the results for the analogous special case in the first-price model,
  the standard bipartite matching problem, which is solvable in
  polynomial time and which has deterministic and randomized online
  algorithms achieving better competitive ratios.  On the positive
  side, we provide a 2-approximation for offline 2PM and a
  5.083-competitive randomized algorithm for online 2PM.  The latter
  result makes use of a new generalization that we prove of a classic
  result on the performance of the ``Ranking'' algorithm for online
  bipartite matching.
\end{abstract}

~\\\\\\\\\\\\\\\\\\

\thispagestyle{empty}
\pagebreak

\setcounter{page}{1}

\newcommand \ourproblem {Second-Price Ad Auctions}
\newcommand \SecondPM {Second-Price Matching}
\section{Introduction}

The rising economic importance of online sponsored search advertising
has led to a great deal of research focused on developing its
theoretical underpinnings.  (See, e.g.,~\cite{Lahaie07} for a survey).
Since search engines such as Google, Yahoo!~and MSN depend on
sponsored search for a significant fraction of their revenue, a key
problem is how to optimally allocate ads to keywords (user searches)
so as to maximize search engine
revenue~\cite{Abrams07,Andelman04,Azar08,Buchbinder07,Chakrabarty08,Goel08a,Goel08b,Mahdian07,Mehta07,Srinivasan08}.
Most of the research on the dynamic version of this problem assumes
that once the participants in each keyword auction are determined, the
pricing is done via a first-price auction; in other words, bidders pay
what they bid. This does not realistically model the standard
mechanism used by search engines, called the Generalized Second Price
mechanism (GSP) \cite{Edelman07,Varian07}.

In an attempt to model reality more closely, we study the {\em
  \ourproblem} problem, which is the analogue of the above allocation
problem when bidders' payments are determined by the GSP mechanism
and there is only one slot for each keyword.  The GSP mechanism for a
given keyword auction reduces to a second-price auction when there is
one slot per keyword -- given the participants in the auction, it
allocates the advertisement slot to the highest bidder, charging that
bidder the bid of the second-highest bidder.

In the \ourproblem\ problem, we assume that there is a set of keywords
$U$ and a set of bidders $V$, where each bidder $v \in V$ has a known
daily budget $B_v$ and a non-negative bid $b_{u,v}$ for every keyword
$u \in U$.  The keywords are ordered by their arrival time, and as
each keyword arrives, the algorithm (i.e., the search engine) chooses
the bidders to participate in that particular auction and runs a
second-price auction with respect to those participants. Thus, instead
of selecting one bidder for each keyword, two bidders need to be
selected by the algorithm. Of these two bidders, the bidder with the
higher bid (where bids are always reduced to the minimum of the actual
bid and the bidders' remaining budget) is allocated that keyword's
advertisement slot at the price of the other bid.

This process results in an allocation and pricing of the advertisement
slots associated with each of the keywords. The goal is to select the
bidders participating in each auction to maximize the total profit
extracted by the algorithm. For an example instance of this problem,
see Figure~\ref{fig:2paa_example}.

We note that for simplicity of presentation we have chosen to present 
the model, the algorithms and the results as if the bidders are competing 
for a single slot. 
Obviously, our hardness results hold for the multi-slot problem as well. 

\begin{figure}[!b]
  \begin{center}
    \begin{tabular}{cccc}
      \subfigure[]{\scalebox{0.7}{\label{fig:2paa_example:a}
          \includegraphics{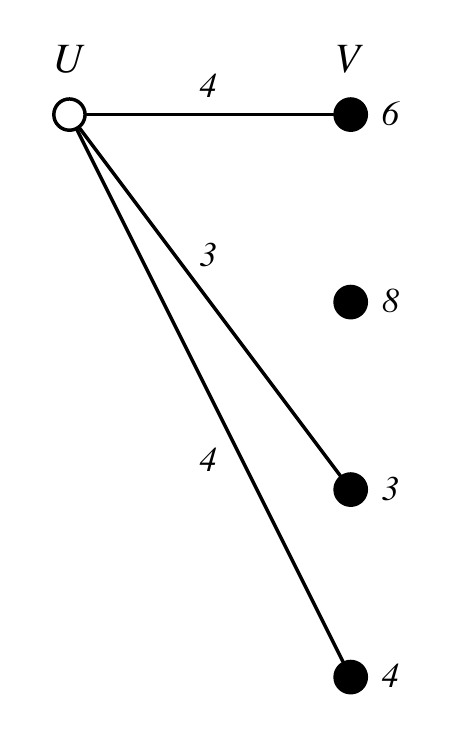}
        } 
      } &
      \subfigure[]{\scalebox{0.7}{\label{fig:2paa_example:b}
          \includegraphics{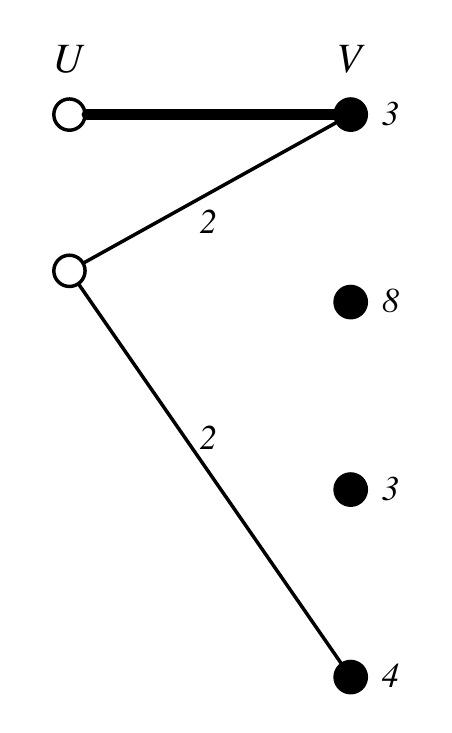}
        }
      } &
      \subfigure[]{\scalebox{0.7}{\label{fig:2paa_example:c}
          \includegraphics{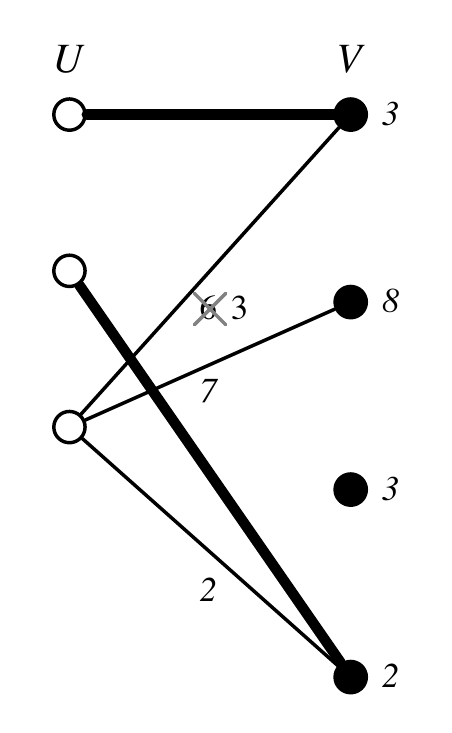}
        } 
      } &
      \subfigure[]{\scalebox{0.7}{\label{fig:2paa_example:d}
          \includegraphics{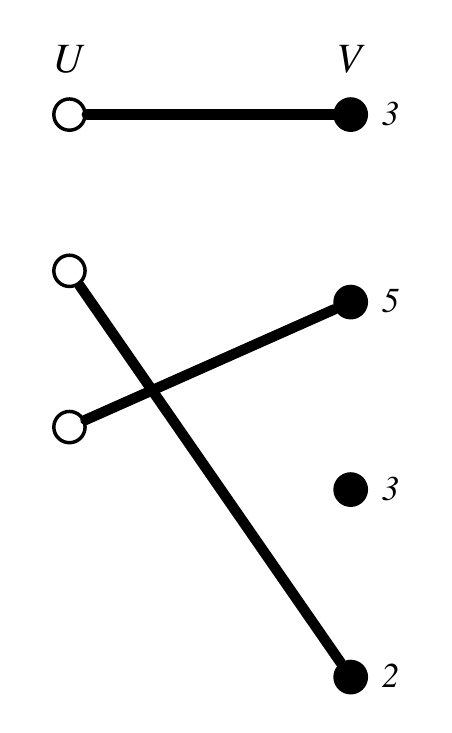}
        } 
      } 
    \end{tabular}
  \end{center}
  \caption{\small An example of the \ourproblem\ problem: The nodes in $U$
    are keywords and the nodes in $V$ are bidders.  The number
    immediately to the right of each bidder represents its remaining
    budget, and the number next to each edge connecting a bidder to a
    keyword represents the bid of that bidder for that keyword.
    Figure~\ref{fig:2paa_example:a} shows the situation when the first
    keyword arrives.  For this keyword, the search engine selects the
    first bidder, whose bid is 4, and the third bidder, whose bid is 3.
    The keyword is allocated to the first bidder at a price of 3,
    thereby reducing that bidder's budget by 3.
    Figure~\ref{fig:2paa_example:b} shows the situation when the
    second keyword arrives. The first and fourth bidders are selected,
    and the keyword is allocated to the fourth bidder at a price of 2,
    thereby reducing its remaining budget to 2. As each keyword
    arrives, the bid of a bidder for that keyword is adjusted to the
    minimum of its original bid and its remaining budget. Thus, for
    example, when the third keyword arrives, as shown in
    Figure~\ref{fig:2paa_example:c}, the bid of the first bidder for
    that keyword is adjusted from its original value of 6 down to 3
    since that is its remaining budget. The two bidders then selected
    are the first and the second, and the keyword is allocated to the
    second bidder at a price of 3.  } \label{fig:2paa_example}
\end{figure}

\subsection{Our Results}

We begin by considering the {\em offline} version of the
\ourproblem\ problem, in which the algorithm knows all of the original
bids of the bidders (Section~\ref{sec:flexible}).  Our main result
here is that it is NP-hard to approximate the optimal solution to this
problem to within a factor better than $m/R_{min}$, where $m$ is the
number of keywords and $R_{min} \geq 1$ is a constant independent of
$m$ such that no bidder bids more than $1/R_{min}$ of its initial
budget on any keyword.  Thus, {\em even when bids are small compared
  to budgets}, it is not possible in the worst case to get a good
approximation to the optimal revenue. (We show that it is trivial to
get a matching approximation algorithm.)  This result stands in sharp
contrast to the standard First-Price Ad Auctions problem, for which
there is a 4/3-approximation to the offline
problem~\cite{Chakrabarty08,Srinivasan08} (even for $R_{min} = 1$),
and an $e/(e-1)$-competitive algorithm to the online problem when bids
are small compared to budgets~\cite{Buchbinder07,Mehta07} (i.e., as
$R_{min} \rightarrow \infty$).

We then turn our attention to a theoretically appealing special case
that we call {\em \SecondPM}.  In this version of the problem, all
bids are either 0 or 1 and all budgets are 1. As before, the keywords
are ordered by arrival time, and a keyword $u$ can be matched to a
bidder $v$ with a profit of 1 only if $b_{u,v}=1$, there is a distinct
bidder $v'$ with $b_{u,v'}=1$, and neither $v$ nor $v'$ has been
matched to a keyword that arrived before $u$.  For an example instance
of this problem, see Figure~\ref{fig:2pm_example}.

\begin{figure}
  \begin{center}
    \begin{tabular}{cccc}
      \subfigure[]{\scalebox{0.7}{\label{fig:2pm_example:a}
          \includegraphics{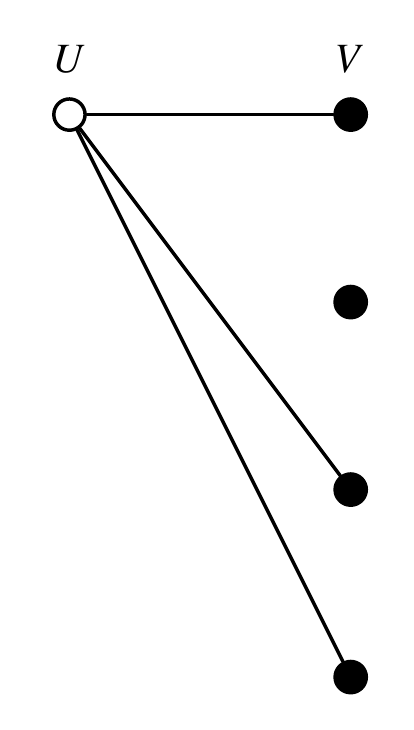}
        } 
      } &
      \subfigure[]{\scalebox{0.7}{\label{fig:2pm_example:b}
          \includegraphics{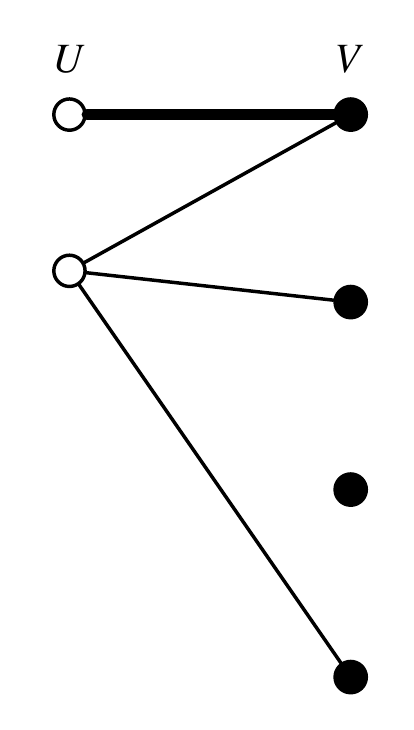}
        }
      } &
      \subfigure[]{\scalebox{0.7}{\label{fig:2pm_example:c}
          \includegraphics{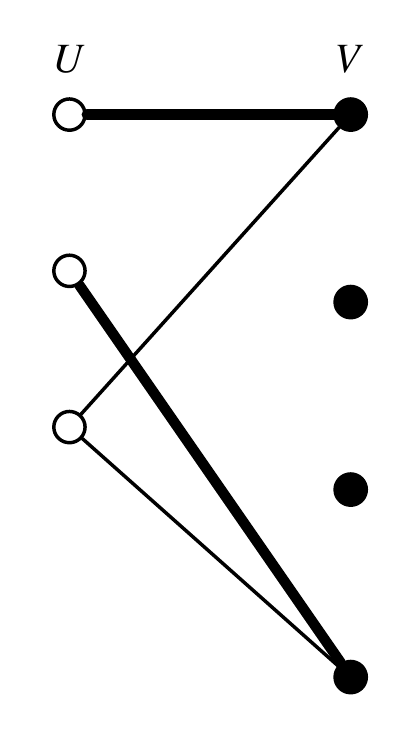}
        } 
      } &
      \subfigure[]{\scalebox{0.7}{\label{fig:2pm_example:d}
          \includegraphics{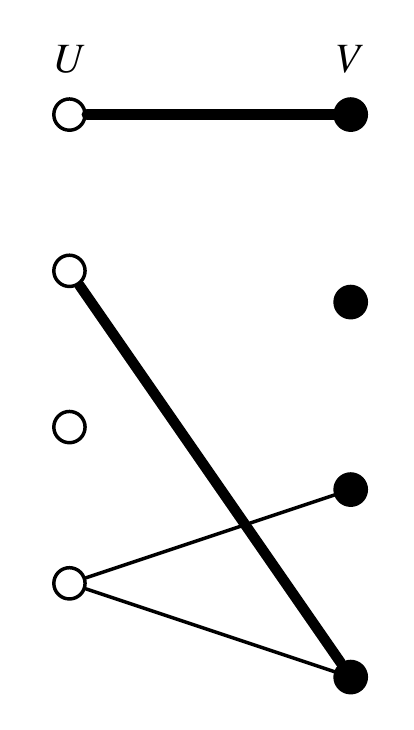}
        } 
      } 
    \end{tabular}
  \end{center}
  \caption{\small An example of the \SecondPM\ problem: The nodes in $U$ are
    keywords, the nodes in $V$ are bidders, and an edge $(u,v)$
    represents the fact that bidder $v$ bids 1 for keyword $u$. The
    budget of each bidder is 1.  In this example, neither of the last
    two keywords can be matched for profit.  In particular, the fourth
    keyword cannot be matched for profit because it has only one
    neighbor that is not already matched (and thus no bidder available
    to act as a second-price bidder).
  } \label{fig:2pm_example}
\end{figure}

Recall that for First-Price Matching or, as we know and love it,
maximum bipartite matching, the offline problem can of course be
solved optimally in polynomial time, whereas for the online problem we
have the trivial 2-competitive deterministic greedy algorithm and the
celebrated $e/(e-1)$-competitive randomized algorithm due to Karp,
Vazirani and Vazirani~\cite{Karp90}, both of which are best possible.

In contrast, we show that the \SecondPM\ problem is APX-hard
(Section~\ref{sec:hardness_offline2pm}).  We also give a
2-approximation algorithm for the offline problem
(Section~\ref{sec:approx_offline2pm}). We then turn to the online
version of the problem. Here, we show that no deterministic online
algorithm can get a competitive ratio better than $m$, where $m$ is
the number of keywords in the instance and that no randomized online
algorithm can get a competitive ratio better than 2
(Section~\ref{sec:lower_online2pm}).  On the other hand, we present a
randomized online algorithm that achieves a competitive ratio of $4
\sqrt{e}/(\sqrt{e}-1) \approx 5.08$
(Section~\ref{sec:rand_upper_online2pm}).  To obtain this competitive
ratio, we prove a generalization of the result due to Karp, Vazirani,
and Vazirani~\cite{Karp90} and Goel and Mehta~\cite{Goel08b} that the
{\em Ranking} algorithm for online bipartite matching achieves a
competitive ratio of $e/(e-1)$.

\subsection{Related Work}

As discussed above, the related First-Price Ad Auctions
problem\footnote{This problem has also been called the {\em Adwords}
  problem~\cite{Mehta07} and the {\em Maximum Budgeted Allocation}
  problem~\cite{Azar08,Chakrabarty08,Srinivasan08}.  It is an
  important special case of SMW
  \cite{Dobzinski06,Feige06b,Khot05,Lehmann06,Mirrokni08,Vondrak08},
  the problem of maximizing utility in a combinatorial auction in
  which the utility functions are submodular, and is also related to
  the Generalized Assignment Problem (GAP)
  \cite{Chekuri00,Feige06b,Fleischer06,Shmoys93}.  } has received a
fair amount of attention.  Mehta et al.~\cite{Mehta07} presented an
algorithm for the online version that achieves an optimal competitive
ratio of $e/(e-1)$ for the case when the bids are much smaller than
the budgets (i.e., when $R_{min} \rightarrow \infty$), a result also
proved by Buchbinder et al.~\cite{Buchbinder07}. When there is no
restriction on the values of the bids relative to the budgets (i.e.,
when $R_{min}$ can be as low as $1$), the best known competitive ratio
is 2~\cite{Lehmann06}.  For the offline version of the problem, a
sequence of
papers~\cite{Lehmann06,Andelman04,Feige06b,Azar08,Srinivasan08,Chakrabarty08}
culminating in a paper by Chakrabarty and Goel, and independently, a
paper by Srinivasan, show that the offline problem can be approximated
to within a factor of $4/(4-1/R_{min})$ and that there is no
polynomial time approximation algorithm that achieves a ratio better
than 16/15 unless $P=NP$~\cite{Chakrabarty08}.

The most closely related papers to this one are the works of Abrams,
Mendelevitch and Tomlin~\cite{Abrams07}, and of Goel, Mahdian,
Nazerzadeh and Saberi~\cite{Goel08a}.  
The latter looks at the online allocation problem
when the search engine is committed to charging under the GSP scheme,
with multiple slots per keyword.  They study two models, both of which
differ from the model studied in this paper, even for one slot.  Their
first model, called the {\em strict} model, allows a bidder's bid to
be above his remaining budget, as long as the remaining budget is
strictly positive.  However, as in our model, when a bidder is
allocated a slot, that bidder is never charged more than his remaining
budget.  Thus, in the strict model, a bidder $v$ with a negligible
amount of remaining budget can keep his bids high indefinitely, and as
long as bidder $v$ is never allocated another slot, this high bid can
determine the prices other bidders pay on many keywords.\footnote{In
  terms of gaming the system, this would be a great way for a bidder
  to potentially force his competitors to pay a lot for slots without
  backing those payments up with budget. This effect is even worse for
  the non-strict model.} Their second model, called the {\em
  non-strict} model, differs from the strict model in that bidders can
keep their bids positive even {\em after} their budget is
depleted. Thus, even after a bidder's budget is depleted, that bidder
can determine the prices that other bidders pay on keywords
indefinitely. However, a bidder is never charged more than his
remaining budget for a slot. Therefore, if a bidder is allocated a
slot after his budget is fully depleted, it gets the slot for free.

Under the assumption that bids are small compared to budgets,
Goel et al.\ present a $e/(e-1)$-competitive algorithm for the non-strict
model and a 3-competitive algorithm for the strict model. They also
show that $OPT_{strict}\le OPT_{non-strict} \le 2 OPT_{strict}$,
where these quantities refer to the optimal offline revenue of the search
engine in their models. Their algorithms build on the linear programming
formulation of Abrams et al.~\cite{Abrams07} for the offline version
of the strict model.

Interestingly, neither of their models or our model completely
dominates the others in terms of the optimal offline revenue. However,
it is fairly easy to show that neither $OPT_{strict}$ nor
$OPT_{non-strict}$ are ever more than a constant factor larger than
the optimal offline revenue from our model (see
Appendix~\ref{appendix:models} for a proof of this).  On the other
hand, the optimal offline revenue in our model can be $\Omega(m)$
times as big as $OPT_{non-strict}$ and $OPT_{strict}$, where $m$ is
the number of keywords, and this holds even when $R_{min}$ is a large
constant $c$ (see Figure~\ref{fig:higher_revenue}).  This is not
surprising, given the strong hardness of approximation in our model,
versus the fact that constant factor approximations are available,
even online, in theirs.

\begin{figure}
  \begin{center}
    \scalebox{0.7}{\includegraphics{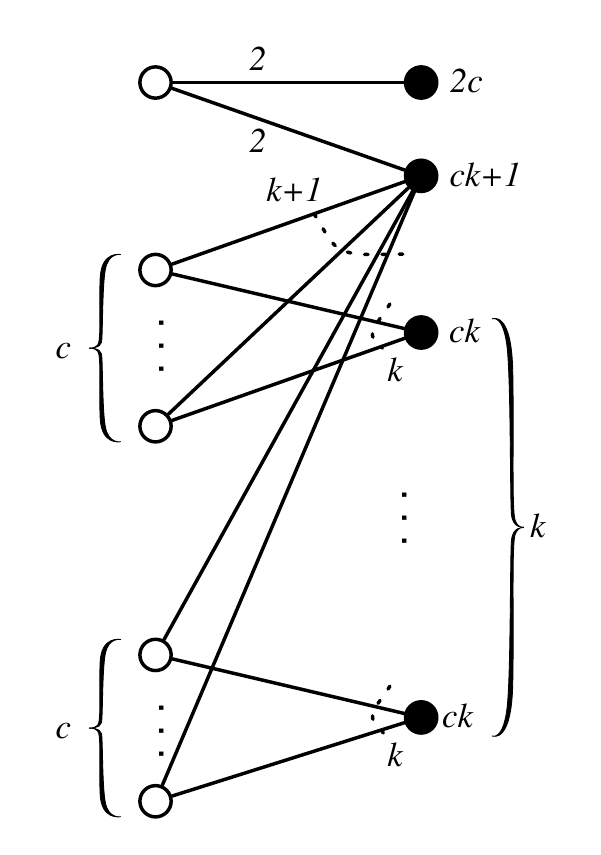}}
  \end{center}
  \caption{\small In this example, $R_{min}$ is equal to a constant $c$,
    i.e., every bid is at most $1/c$ of the budget of the bidder.  In
    the strict model, all keywords except the first must be allocated
    to the second bidder at a price of $k$ (or the remaining budget if
    it's smaller).  Thus, the total profit on this input for the
    strict model is at most $ck+3$.  On the other hand, in our model,
    if we allocate the first keyword to the second bidder, and then
    the next $c-1$ keywords to the second bidder, that bidder's budget
    is reduced to $k-1$. Thus, all of the remaining keywords can be
    allocated to the lower bidder at a price of $k-1$, for a total
    revenue exceeding $ck(k-1)$. For $k$ large, this ratio is $\Omega
    (k) = \Omega(m)$.  }\label{fig:higher_revenue}
\end{figure}

Notice also, that for the special case of \SecondPM\,
the strict model and our model are identical, whereas the non-strict
model reduces to standard maximum matching (on all those keywords which
have at least two bidders bidding on them).

We feel that our model, in which
bidders are not allowed to bid more than their remaining budget,
is more natural then the strict and non-strict
models. It seems inherently unfair that a bidder with negligible or no
budget should be able to indefinitely set high prices for other bidders.
We recognize of course that there may be technical issues
involving the scale and distributed nature of search engines that
could  make it difficult to implement our model precisely.

\section{Model and Notation}\label{sec:model}

We define the Second-Price Ad Auctions (2PAA) problem formally as
follows.  The input is a set of ordered keywords $U$ and bidders $V$.
Each bidder $v \in V$ has a budget $B_v$ and a nonnegative bid
$b_{u,v}$ for every keyword $u \in U$.  We assume that all of bidder
$v$'s bids $b_{u,v}$ are less than or equal to $B_v$.

Let $B_v(t)$ be the remaining budget of bidder $v$ immediately after
the $t$-th keyword is processed (so $B_v(0)= B_v$ for all $v$), and
let $b_{u,v}(t) = \min (b_{u,v}, B_v(t))$. (Both quantities are
defined inductively.)  A solution (or {\em second-price matching}) to
2PAA chooses for the $t$-th keyword $u$ a pair of bidders $v_1$ and
$v_2$ such that $b_{u,v_1}(t-1) \ge b_{u,v_2}(t-1)$, allocates the
slot for keyword $u$ to bidder $v_1$ and charges bidder $v_1$ a price
of $p(t)= b_{u,v_2}(t-1)$, the bid of $v_2$.  (We say that $v_1$ acts
as the \emph{first-price bidder} for $u$ and $v_2$ acts as the
\emph{second-price bidder} for $u$.)  The budget of $v_1$ is then
reduced by $p(t)$, so $B_{v_1}(t) = B_{v_1}(t-1) - p(t)$. For all
other bidders $v \ne v_1$, $B_{v}(t) = B_v(t-1)$.  The final value of
the solution is $\sum_t p(t)$, and the goal is to find a feasible
solution of maximum value.

In the offline version of the problem, all of the bids are known to
the algorithm beforehand, whereas in the online version of the
problem, keyword $u$ and the bids $b_{u,v}$ for each $v \in V$ are
revealed only when keyword $u$ arrives, at which point the algorithm
must irrevocably map $u$ to a pair of bidders without knowing the bids
for the keywords that will arrive later.

The special case referred to as Second-Price Matching (2PM) is the
special case where $b_{u,v}$ is either 0 or 1 for all $(u,v)$ pairs
and $B_v=1$ for all $v$. Perhaps, however, it is more intuitive to
think of it as a variant on maximum bipartite matching.  Viewing it
this way, the input is a bipartite graph $G = (U \cup V, E)$, and the
vertices in $U$ must be matched, in order, to the vertices in $V$ such
that the profit of matching $u \in U$ to $v \in V$ is $1$ if and only
if there is at least one additional vertex $v' \in V$ that is a
neighbor of $u$ and is unmatched at that time.

Note that in 2PM, a keyword can only be allocated for profit if its
degree is at least two.  Therefore, we assume without loss of
generality that for all inputs of 2PM, the degree of every keyword is
is at least two.

For an input to 2PAA, let $R_{min} = \min_{u,v} B_v / b_{u,v}$,
and let $m = |U|$ be the number of keywords.

\section{Hardness of Approximation of 2PAA}\label{sec:flexible}

In this section, we show that it is NP-hard to approximate 2PAA to a
factor better than $m/R_{min}$ for any constant $R_{min}$ independent
of $m$ (Theorem~\ref{thm:hardness_2paa}) and provide a trivial
algorithm with an approximation guarantee that meets this factor.

For a constant $c \geq 1$, let 2PAA($c$) be the version of 2PAA in
which we are promised that $R_{min} \geq c$.  The proof of the
following theorem (presented in full in
Appendix~\ref{appendix:hardness_2paa}) uses a gadget similar to the
instance presented in Figure~\ref{fig:higher_revenue}.  This gadget
exploits the sensitivity of the problem to small changes in budget.
In particular, if the budget of a key bidder is reduced by a small
amount, then the optimal revenue is much higher than if it is not
reduced.  We prove our inapproximability result by setting up the
reduction so that the budget of this bidder is reduced if and only if
the answer to the problem we reduce from is ``yes.''
\begin{theorem}\label{thm:hardness_2paa}
Let $c \geq 1$ be a constant integer.  For any constant $c' > c$, it
is NP-hard to approximate 2PAA($c$) to a factor of $m/c'$.
\end{theorem}
\begin{proof}
See Appendix~\ref{appendix:hardness_2paa}.
\end{proof}

\begin{theorem}\label{thm:matching_2paa}
Let $c \geq 1$ be a constant integer. There is an $m/c$-approximation
to 2PAA($c$).
\end{theorem}
\begin{proof}
For each keyword $u \in U$, let $s_u$ be the second-highest bid for
$u$.  Consider the algorithm that selects the $c$ keywords with the
highest values of $s_u$ and then allocates these keywords to get $s_u$
for each of them (i.e., chooses the two highest bidders for $u$).
Since no bidder bids more than $1/c$ of its budget for any keyword, no
bids are reduced from their original values during this allocation.
Hence, the profit of this allocation is at least $(c/m) \sum_{u \in U}
s_u$.  Since the value of the optimal solution cannot be larger than
$\sum_{u\in U} s_u$, it follows that this is an $m/c$-approximation to
2PAA($c$).
\end{proof}

\section{Offline Second-Price Matching}\label{sec:offline2pm}

In this section, we turn our attention to the offline version of the
special case of Second-Price Matching (2PM).  First, in
Section~\ref{sec:hardness_offline2pm}, we show that 2PM is APX-hard.
This stands in contrast to the maximum matching problem, the
corresponding special case of the First-Price Ad Auctions problem,
which is solvable in polynomial time.  Second, in
Section~\ref{sec:approx_offline2pm}, we give a 2-approximation for
2PM.

\subsection{Hardness of Approximation}\label{sec:hardness_offline2pm}

To prove our hardness result for 2PM, we use the following
result due to Chleb\'{i}k and Chelb\'{i}kov\'{a}.
\begin{theorem}[Chleb\'{i}k and Chelb\'{i}kov\'{a} \cite{Chlebik06}]\label{thm:hardness_vc}
  It is NP-hard to approximate Vertex Cover on 4-regular graphs to
  within $53/52$.
\end{theorem}
The precise statement of our hardness result is the following
theorem.
\begin{theorem}\label{thm:hardness_2pm}
  It is NP-hard to approximate 2PM to within a factor of $364/363$.
\end{theorem}
\begin{proof}
  See Appendix~\ref{appendix:hardness_2pm}
\end{proof}

\subsection{A 2-Approximation Algorithm}\label{sec:approx_offline2pm}

Consider an instance $G = (U \cup V,E)$ of the 2PM problem.  We
provide an algorithm that first finds a maximum matching $f : U
\rightarrow V$ and then uses $f$ to return a second-price matching
that contains at least half of the keywords matched by
$f$.\footnote{Note that $f$ is a partial function.}  Given a matching
$f$, call an edge $(u,v) \in E$ such that $f(u) \not= v$ an {\em
  up-edge} if $v$ is matched by $f$ and $f^{-1}(v)$ arrives before
$u$, and a {\em down-edge} otherwise.  Recall that we have assumed
without loss of generality that the degree of every keyword in $U$ is
at least two.  Therefore, every keyword $u \in U$ that is matched by
$f$ must have at least one up-edge or down-edge.
Theorem~\ref{thm:2_approx} shows that the following algorithm, called
ReverseMatch, is a $2$-approximation for 2PM.

\begin{center}
\fbox{
\begin{tabular}{l}
{\bf ReverseMatch Algorithm:}\\
\hline
{\em Initialization:} \\
Find an arbitrary maximum matching $f:U\rightarrow V$ on $G$. \\
\hline
  {\em Constructing a 2nd-price matching:} \\
  Consider the matched keywords in reverse order of their arrival. \\
  For each keyword $u$: \\
  \hspace{.2in} If keyword $u$ is adjacent to a down-edge $(u,v)$: \\
    \hspace{.4in} Assign keyword $u$ to bidder $f(u)$ (with $v$ acting as the second-price bidder). \\
  \hspace{.2in} Else: \\
    \hspace{.4in} Choose an arbitrary bidder $v$ that is adjacent to keyword $u$. \\
    \hspace{.4in} Remove the edge $(f^{-1}(v),v)$ from $f$. \\
    \hspace{.4in} Assign keyword $u$ to bidder $f(u)$ (with $v$ acting as the second-price bidder). 
\end{tabular}
}
\end{center}

\begin{theorem}\label{thm:2_approx}
 The ReverseMatch algorithm is a 2-approximation.
\end{theorem}
\begin{proof}
  See Appendix~\ref{appendix:2_approx}.
\end{proof}

\section{Online Second-Price Matching}\label{sec:online2pm}

In this section, we consider the online 2PM problem, in which the
keywords arrive one-by-one and must be matched by the algorithm as
they arrive.  We start, in Section~\ref{sec:lower_online2pm}, by giving
a simple lower bound showing that no deterministic algorithm can
achieve a competitive ratio better than $m$, the number of keywords.
Then we move to randomized online algorithms and show that no
randomized algorithm can achieve a competitive ratio better than $2$.
In Section~\ref{sec:rand_upper_online2pm}, we provide a randomized
online algorithm that achieves a competitive ratio of
$4\sqrt{e}/(\sqrt{e}-1) \approx 5.083$.

\subsection{Lower Bounds}\label{sec:lower_online2pm}

The following theorem establishes our lower bound on deterministic
algorithms, which matches the trivial algorithm of arbitrarily
allocating the first keyword to arrive, and refusing to allocate any
of the remaining keywords.
\begin{theorem}
For any $m$, there is an adversary that creates a graph with $m$
keywords that forces any deterministic algorithm to get a competitive
ratio no better than $1/m$.
\end{theorem}
\begin{proof}
The adversary shows the algorithm a single keyword (keyword $1$) that
has two adjacent bidders, $a_1$ and $b_2$.  If the algorithm does not
match keyword $1$ at all, a new keyword $2$ arrives that is adjacent
to two new bidders $a_2$ and $b_2$.  The adversary continues in this
way until either $m$ keywords arrive or the algorithm matches a
keyword $k < m$.  In the first case, the algorithm's performance is at
most $1$ (because it might match keyword $m$), whereas the adversary
can match all $m$ keywords.  Hence, the ratio is at most $1/m$.

In the second case, the adversary continues as follows.  Suppose
without loss of generality that the algorithm matches keyword $k$ to
$a_k$.  Then each keyword $i$, for $k+1 \leq i \leq m$, has one edge
to $a_k$ and one edge to a new bidder $c_k$.  Since the algorithm
cannot match any of these keywords for a profit, its performance is
$1$.  The adversary can clearly match each keyword $i$ for profit, for
$1 \leq i \leq k -1$, and if it matches keyword $k$ to $b_k$, then it
can use $a_k$ as a second-price bidder for the remaining keywords to
match them all to the $c_i$'s for profit.  Hence, the adversary can
construct a second-price matching of size at least $m$.
\end{proof}

The following theorem establishes our lower bound of $2$ for any
(randomized) online algorithm for 2PM.
\begin{theorem}\label{thm:rand_lower_online2pm}
The competitive ratio of any randomized algorithm for 2PM must
be at least $2$.
\end{theorem}
\begin{proof}
  See Appendix~\ref{appendix:rand_lower_online2pm}.
\end{proof}

\subsection{A Randomized Competitive Algorithm}\label{sec:rand_upper_online2pm}

In this section, we provide an algorithm that achieves a competitive
ratio of $2\sqrt{e}/(\sqrt{e} - 1) \approx 5.083$.  The result builds
on a new generalization of the result that the Ranking algorithm for
online bipartite matching achieves a competitive ratio of $e/(e-1)
\approx 1.582$.  This was originally shown by Karp, Vazirani, and
Vazirani \cite{Karp90}, though a mistake was recently found in their
proof by Krohn and Varadarajan and corrected by Goel and Mehta
\cite{Goel08b}.

The online bipartite matching problem is merely the first-price
version of 2PM, i.e., the problem in which there is no requirement for
there to exist a second-price bidder to get a profit of $1$ for a
match.  The Ranking algorithm chooses a random permutation on the
bidders $V$ and uses that to choose matches for the keywords $U$
as they arrive.  This is described more precisely below.

\begin{center}
\fbox{
\begin{tabular}{l}
{\bf Ranking Algorithm:}\\
\hline
Initialization: \\
Choose a random permutation (ranking) $\sigma$ of the bidders $V$.\\
\hline
  Online Matching: \\
  Upon arrival of keyword $u \in U$: \\
  \hspace{.2in} Let $N(u)$ be the set of neighbors of $u$ that have not been matched yet. \\
  \hspace{.2in} If $N(u) \neq \emptyset$, match $u$ to the bidder $v \in N(u)$ that minimizes $\sigma(v)$.
\end{tabular}
}
\end{center}

Karp, Vazirani, and Vazirani, and Goel and Mehta prove the following
result.
\begin{theorem}[Karp, Vazirani, and Vazirani \cite{Karp90}
    and Goel and Mehta \cite{Goel08b}]\label{thm:old_ranking} 
  The Ranking algorithm for online bipartite matching achieves a
  competitive ratio of $e/(e-1) + o(1)$.
\end{theorem}
In order to state our generalization of this result, we define
the notion of a $\emph{left $k$-copy}$ of a bipartite graph $G = (U
\cup V, E)$.  Intuitively, a left $k$-copy of $G$ makes $k$ copies of
each keyword $u \in U$ such that the neighborhood of a copy of $u$ is
the same as the neighborhood of $u$.  More precisely, we have the
following definition.
\begin{definition}\label{def:kcopy}
Given a bipartite graph $G = (U_G \cup V, E_G)$, define a {\bf left
  $k$-copy} of $G$ to be a graph $H = (U_H \cup V,E_H)$ for which
$|U_H| = k|U_G|$ and for which there exists a map $\zeta: U_H
\rightarrow U_G$ such that
\begin{itemize}
\item
for each $u_G \in U_G$ there are exactly $k$ vertices $u_H \in U_H$
such that $\zeta(u_H) = u_G$, and

\item
for all $u_H \in U_H$ and $v \in V$, $(u_H, v) \in E_H$ if and
only if $(\zeta(u_H), v) \in E_G$.
\end{itemize}
\end{definition}

Our generalization of Theorem~\ref{thm:old_ranking} describes the
competitive ratio of Ranking on a graph $H$ that is a $k$-copy of $G$.
It's proof, presented in Appendix~\ref{appendix:ranking}, is a 
modification to the proof of Theorem~\ref{thm:old_ranking} presented
by Birnbaum and Mathieu~\cite{Birnbaum08}.
\begin{theorem}\label{thm:ranking}
Let $G = (U_G \cup V, E_G)$ be a bipartite graph that has a maximum
matching of size $OPT_{1P}$, and let $H = (U_H \cup V, E_H)$ be a left
$k$-copy of $G$.  Then the expected size of the matching returned by
Ranking on $H$ is at least
\begin{equation*}
k OPT_{1P} \paren{1 - \frac{1}{e^{1/k}} + o(1)} \enspace .
\end{equation*}
\end{theorem}
\begin{proof}
See Appendix~\ref{appendix:ranking}.
\end{proof}

Using this result, we are able to prove that the following algorithm,
called RankingSimulate, achieves a competitive ratio of
$4\sqrt{e}/(\sqrt{e}-1)$.
\begin{center}
\fbox{
\begin{tabular}{l}
{\bf RankingSimulate Algorithm:}\\
\hline
\emph{Initialization:} \\
Set $M$, the set of \emph{matched} bidders, to $\emptyset$. \\
Set $R$, the set of \emph{reserved} bidders, to $\emptyset$. \\
Choose a random permutation (ranking) $\sigma$ of the bidders $V$.\\
\hline
  \emph{Online Matching:} \\
  Upon arrival of keyword $u \in U$: \\  
  \hspace{.2in} Let $N(u)$ be the set of neighbors of 
  $u$ that are not in $M$ or $R$. \\
  \hspace{.2in} If $N(u) = \emptyset$, do nothing. \\
  \hspace{.2in} If $|N(u)| = 1$, let $v$ be the single bidder in $N(u)$. \\
  \hspace{.4in} With probability $1/2$, match $u$ to $v$ and add $v$
  to $M$, and \\
  \hspace{.4in} With probability $1/2$, add $v$ to $R$. \\
  \hspace{.2in} If $|N(u)| \geq 2$, let $v_1$ and $v_2$ be the two
  distinct bidders in $N(u)$ that minimize $\sigma(v)$. \\
  \hspace{.4in} With probability $1/2$, match $u$ to $v_1$, add
  $v_1$ to $M$, and add $v_2$ to $R$, and \\
  \hspace{.4in} With probability $1/2$, match $u$ to $v_2$, add
  $v_1$ to $R$, and add $v_2$ to $M$.
\end{tabular}
}
\end{center}

Let $G = (U_G \cup V, E_G)$ be the bipartite input graph to 2PM, and let
$H = (U_H \cup V, E_H)$ be a left $2$-copy of $H$.  In the arrival order
for $H$, the two copies of each keyword $u_G \in U$ arrive in
sequential order.

\begin{lemma}\label{lem:relateranking}
Fix a ranking $\sigma$ on $V$.  For each bidder $v \in V$, let $X_v$
be the indicator variable for the event that $v$ is matched by Ranking
on $H$, when the ranking is $\sigma$.\footnote{Note that once $\sigma$
  is fixed, $X_v$ is deterministic.}  Let $X'_v$ be the indicator
variable for the event that $v$ is matched by RankingSimulate on $G$,
when the ranking is $\sigma$.  Then $\expect(X'_v) = X_v / 2$.
\end{lemma}
\begin{proof}
It is easy to establish the invariant that for all $v \in V$, $X_v =
1$ if and only if RankingSimulate puts $v$ in either $M$ or $R$.
Furthermore, each bidder $v \in V$ is put in $M$ or $R$ at most once
by RankingSimulate.  The lemma follows because each time
RankingSimulate adds a bidder $v$ to $M$ or $R$, it matches it with
probability $1/2$.
\end{proof}

\begin{theorem}
The competitive ratio of RankingSimulate is $2\sqrt{e}/(\sqrt{e} - 1)
\approx 5.083$.
\end{theorem}
\begin{proof}
For a permutation $\sigma$ on $V$, let $\rankingsim(\sigma)$ be the
matching of $G$ returned by RankingSimulate, and let
$\ranking(\sigma)$ be the matching of $H$ returned by Ranking.
Lemma~\ref{lem:relateranking} implies that, conditioned on $\sigma$,
$\expect(|\rankingsim(\sigma)|) = |\ranking(\sigma)| / 2$.  
By
Theorem~\ref{thm:ranking},
\begin{equation*}
\expect (|\rankingsim(\sigma)|)
= \frac{1}{2} \expect (|\ranking(\sigma)|)
\geq  OPT_{1P} \paren{1 - 1/e^{1/2} + o(1)} \enspace .
\end{equation*}

Fix a bidder $v \in V$.  Let $P_v$ be the profit from $v$ obtained by
RankingSimulate.  Suppose that $v$ is matched by RankingSimulate to
keyword $u \in U_G$.  Recall that we have assumed without loss of
generality that the degree of $u$ is at least $2$.  Let $v' \ne v$ be
another bidder adjacent to $u$.  Then, given that $v$ is matched to
$u$, the probability that $v'$ is matched to any keyword is no greater
than $1/2$.  Therefore, $\expect(P_v|\mbox{$v$ matched}) \geq 1/2$.
Hence, the expected value of the second-price matching returned by
RankingSimulate is
\begin{eqnarray*}
\sum_{v \in V} \expect(P_v) & = &
\sum_{v \in V} \expect(P_v | \mbox{$v$ matched}) \pr(\mbox{$v$ matched}) \\
& \geq & \frac12 \sum_{v \in V} \pr(\mbox{$v$ matched}) \\
& =    & \frac12 \expect(|\rankingsim(\sigma)|) \\
& \geq & \frac12 OPT_{1P} \paren{1 - 1/e^{1/2} + o(1)} \\
& \geq & \frac12 OPT_{2P} \paren{1 - 1/e^{1/2} + o(1)} \enspace ,
\end{eqnarray*}
where $OPT_{2P}$ is the size of the optimal second-price matching on $G$.
\end{proof}

\section{Conclusion}

In this paper, we have shown that the complexity of the Second-Price
Ad Auctions problem is quite different from that of the more studied
First-Price Ad Auctions problem, and that this discrepancy extends to
the special case of 2PM, whose first-price analogue is bipartite
matching.  On the positive side, we have given a 2-approximation
for offline 2PM and a 5.083-competitive algorithm for online 2PM.

Some open questions remain.  Closing the gap between $2$ and $364/363$
in the approximability of offline 2PM is one clear direction for
future research, as is closing the gap between $2$ and $5.083$ in the
competitive ratio for online 2PM.  Another question we leave open is
whether the analysis for RankingSimulate is tight, though we expect
that it is not.  


\pagebreak
\bibliographystyle{plain}
\bibliography{thinking_twice}

\appendix
\section{Discussion of Related Models}\label{appendix:models}
In this section, we prove the statements claimed in the introduction
regarding the relationships between the optimal solutions of our model
and the strict and non-strict models of Goel et al.~\cite{Goel08a}.
Given an instance $A$, let $OPT_{2P}$ be the optimal solution value
in our model; let $OPT_{strict}$ be the optimal solution value under
the strict model; and let $OPT_{non-strict}$ be the optimal solution
value under the non-strict model.  We prove the following theorem.
\begin{theorem}\label{thm:relationship}
For any instance $A$, $OPT_{non-strict} \leq (2+1/R_{max})OPT_{strict}
\leq 8(2+1/R_{max}) OPT_{2P}$.
\end{theorem}
The first inequality follows from the work of Goel et~al.~\cite{Goel08a},
so we need only prove that $OPT_{strict} \leq 8 OPT_{2P}$.

The core of our argument is a reduction from 2PAA to the First-Price
Ad Auctions problem (1PAA),\footnote{Recall that this problem has also
  been called the {\em Adwords} problem~\cite{Mehta07} and the {\em
    Maximum Budgeted Allocation}
  problem~\cite{Azar08,Chakrabarty08,Srinivasan08}.} in which only one
bidder is chosen for each keyword and that bidder pays the minimum of
its bid and its remaining budget.  Given an instance $A$ of 2PAA, we
construct an instance $A'$ of 1PAA problem by replacing each bid
$b_{u,v}$ by
\begin{equation*}
  b'_{u,v} \triangleq \max_{v' \ne v~:~ b_{u,v'} \leq b_{u,v}} b_{u,v'}
  \enspace .
\end{equation*}
Denote by $OPT_{1P}(A')$ the optimal value of the first-price model on
$A'$. The following two lemmas prove Theorem~\ref{thm:relationship} by
relating both $OPT_{non-strict}(A)$ and $OPT_{2P}(A)$ to
$OPT_{1P}(A')$.

\begin{lemma}
 $OPT_{non-strict}(A) \leq OPT_{1P}(A')$.
\end{lemma}
\begin{proof}
  For an instance $A$, we can view a non-strict second-price
  allocation of $A$ as a pair of (partial) functions $f_1$ and $f_2$
  from the keywords $U$ to the bidders $V$, where $f_1$ maps each
  keyword to the bidder to which it is allocated and $f_2$ maps each
  keyword to the bidder acting as its second-price bidder.  Thus, if
  $f_1(u) = v$ and $f_2(u) = v'$ then $u$ is allocated to $v$ and the
  price $v$ pays is $b_{u,v'}$.  We have, for all such $u$, $v$, and
  $v'$, that $b_{u,v'} \leq b'_{u,v}$.

  We construct the first-price allocation on $A'$ defined by $f_1$ and
  claim that the value of this first-price allocation is at least the
  value of the non-strict allocation defined by $f_1$ and $f_2$. It
  suffices to show that for any bidder $v$, the profit that the
  non-strict allocation gets from $v$ is at most the profit that the
  first-price allocation gets from $v$, or in other words,
  $$ 
  \min\left(B_v, \sum_{u: f_1(u) = v}
  b_{u,f_2(u)}\right) \leq \min\left(B_v, \sum_{u: f_1(u)
    = v} b'_{u,v}\right) \enspace .
  $$ 
  This inequality follows trivially from the fact that
  $b_{u,f_2(u)} \leq b'_{u,f_1(u)}$ for all allocated keywords $u$,
  and hence the lemma follows.
\end{proof}

\begin{lemma}
  $OPT_{1P}(A') \leq 8OPT_{2P}(A)$.
\end{lemma}
\begin{proof}
  Given an optimal first-price allocation of $A'$, we can assume
  without loss of generality that each bidder's budget can only be
  exhausted by the last keyword allocated to it, or, more formally, if
  $u_1, u_2, \ldots u_k$ are the keywords that are allocated to a
  bidder $v$ and they come in that order, then we can assume that
  $\sum_{i=1}^{k-1} b'_{u_i,v} < B_v$.  The reason for this is that if
  for some $j < k$, $\sum_{i=1}^{j-1} b'_{u_i,v} < B_v$ and
  $\sum_{i=1}^j b'_{u_i,v} \geq B_v$, then we can ignore the allocation
  of $u_{j+1}, \ldots u_k$ to $v$ without losing any profit.
 
  With this assumption, we design a randomized algorithm that
  constructs a second-price allocation on $A$ whose expected value in
  our model is at least $1/8$ of the first-price allocation's value.
  Viewing the first-price allocation of $A'$ as a (partial) function
  $f$ from the keywords $U$ to the bidders $V$ and denoting by
  $s(u,v)$ the bidder $v'$ for which $b_{v'u} = b'_{vu}$, the
  algorithm is as follows.

\begin{center}
\fbox{
\begin{tabular}{l}
{\bf Random Construction:}\\
\hline
  Randomly mark each bidder with probability $1/2$. \\
  For each unmarked bidder $v$: \\
    \hspace{.2in} Let $S_v = \emptyset$.\\
    \hspace{.2in} For each keyword $u$ such that $f(u) = v$:\\
      \hspace{.4in} If $s(u,v)$ is marked: $S_v = S_v \cup \{u\}$.\\
    \hspace{.2in} Assume that $S_v = \{u_1, u_2, \ldots u_k\}$, where $u_1, u_2, \ldots u_k$ come in that order.\\
      \hspace{.4in} If $\sum_{i=1}^k b'_{u_i,v} \leq B_v$:\\
        \hspace{.6in} Let $f_1(u_i) = v$ and $f_2(u_i) = s(u_i,v)$ for all $i \leq k$.\\
      \hspace{.4in} Else: \\
        \hspace{.6in} If $\sum_{i=1}^{k-1} b'_{u_i,v} \geq b'_{u_k,v}$: 
          let $f_1(u_i) = v$ and $f_2(u_i) = s(u_i,v)$ for all $i \leq k-1$. \\
        \hspace{.6in} Else: let $f_1(u_k) = v$ and $f_2(u_k) = s(u_k, v)$.\\
\end{tabular}
}
\end{center}

We claim that for the $f_1$ and $f_2$ defined by this construction,
whenever $f(u_i)$ is set to $v$, the profit from that allocation is
$b'_{u_i,v}$.  This is not trivial because in our model, if a bidder's
remaining budget is smaller than its bid for a keyword, it changes its
bid for that keyword to its remaining budget. However, one can easily
verify that in all cases, if we set $f_1(u_i) = v$ and $f_2(u_i) =
s(u_i, v)$, the remaining budget of $v$ is at least $b'_{u_i,v} =
b_{u_i,s(u_i,v)}$. Thus, the (modified) bid of $f_1(u_i)$ for $u_i$ is
still at least the original bid of $f_2(u_i)$ for $u_i$.

We claim that the expected value of the second-price allocation
defined by $f_1$ and $f_2$ is at least $1/8 OPT_{1P}(A')$. For each
bidder $v$, let $X_v$ be the random variable denoting the profit that
$f_1$ and $f_2$ get from $v$, and let $Y_v$ be the profit that $f$
gets from $v$. We have $OPT_{1P}(A') = \sum_v Y_v$, so it suffices to
show that $E(X_v) \geq 1/8Y_v$ for all $v \in V$.

Consider any $v \in V$ that is unmarked. Let $T_v = \{u: f(u) =
v\}$. If $\sum_{u \in S_v} b'_{u,v} \leq B_v$ then $X_v = \sum_{u \in
  S_v} b'_{u,v}$. If $\sum_{u \in S_v} b'_{u,v} > B_v$ then $X_v \geq
\sum_{u \in S_v} b'_{u,v}/2$. Thus, in both case, we have
$$
  E[X_v|v\ \textrm{is unmarked}] \geq E[\sum_{u \in S_v} b'_{u,v}/2 | v\ \textrm{is unmarked}] = \sum_{u \in T_v} b'_{u,v}/4 = Y_v/4 \enspace ,
$$
which implies
$$
  E[X_v] \geq E[X_v|v\ \textrm{is unmarked}]Pr[v\ \textrm{is unmarked}] = 1/2\cdot Y_v/4 = Y_v/8 \enspace .
$$
\end{proof}

\section{Proof of Theorem~\ref{thm:ranking}}\label{appendix:ranking}
In this appendix, we provide a full proof of
Theorem~\ref{thm:ranking}.  The proof presented here is quite similar
to the simplified proof of Theorem~\ref{thm:old_ranking} presented by
Birnbaum and Mathieu~\cite{Birnbaum08}.  For intuition into the proof
presented here, the interested reader is referred to that
work.\footnote{For those familiar with the proof in~\cite{Birnbaum08},
  the main difference between the proof of Theorem~\ref{thm:ranking}
  presented here and the proof of Theorem~\ref{thm:old_ranking}
  presented in \cite{Birnbaum08} appears in
  Lemma~\ref{lem:rankingmain}.  Instead of letting $u$ be the single
  vertex that is matched to $v$ by the perfect matching, as is done in
  \cite{Birnbaum08}, we choose $u$ uniformly at random from one of the
  $k$ vertices that correspond to the vertex that is matched to $v$ by
  the perfect matching.  The rest of the proof is essentially the
  same, but we present its entirety here for completeness.}

Let $G = (U_G \cup V, E_G)$ be a bipartite graph and let $H = (U_H
\cup V, E_H)$ be a left $k$-copy of $G$. Let $\zeta : U_H \rightarrow
U_G$ be a map that satisfies the conditions of
Definition~\ref{def:kcopy}.  Let $M_G \subseteq E_G$ be a maximum
matching of $G$.

Let $\ranking(H, \pi, \sigma)$ denote the matching constructed on $H$
for arrival order $\pi$, when the ranking is $\sigma$.  Consider
another process in which the vertices in $V$ arrive in the order given
by $\sigma$ and are matched to the available vertex $u \in U_H$ that
minimizes $\pi(u)$.  Call the matching constructed by this process
$\ranking'(H, \pi, \sigma)$.  
It is not hard to see that these matchings are identical, a fact
that is proved in \cite{Karp90}.
\begin{lemma}[Karp, Vazirani, and Vazirani \cite{Karp90}]\label{lem:duality}
For any permutations $\pi$ and $\sigma$, $\ranking(H, \pi, \sigma) =
\ranking'(H, \pi,\sigma)$.\end{lemma}

The following monotonicity lemma shows that removing vertices in $H$
can only decrease the size of the matching returned by Ranking.

\begin{lemma}\label{lem:monotonicity}
Let $\pi_H$ be an arrival order for the vertices in $U_H$, and let
$\sigma_H$ be a ranking on the vertices in $V$.  Suppose that $x$ is a
vertex in $U_H \cup V$, and let $H' = (U_{H'}, V_{H'}, E_{H'}) = H
\setminus \braces{x}$.  Let $\pi_{H'}$ and $\sigma_{H'}$ be the
orderings of $U_{H'}$ and $V_{H'}$ induced by $\pi_H$ and $\sigma_H$,
respectively.  Then $\ranking(H', \pi_{H'}, \sigma_{H'}) \leq
\ranking(H, \pi_H, \sigma_H)$.
\end{lemma}
\begin{proof}
Suppose first that $x \in U_H$.  In this case, $V = V_{H'}$ and
$\sigma_{H} = \sigma_{H'}$.  Let $Q_t(H) \subseteq V$ be the set of
vertices matched to vertices in $U_H$ that arrive at or before time
$t$ (under arrival order $\pi_H$ and ranking $\sigma_H$), and let
$Q_t(H') \subseteq V$ be the set of vertices matched to vertices in
$U_{H'}$ that arrive at or before time $t$ (under arrival order
$\pi_{H'}$ and ranking $\sigma_H$).  We prove by induction on $t$ that
$Q_{t-1}(H') \subseteq Q_{t}(H)$, which by substituting $t = n$ is
sufficient to prove the claim.  The statement holds when $t = 1$,
since $Q_0(H') = \emptyset$.  Now supposing we have $Q_{t-2}(H')
\subseteq Q_{t-1}(H)$, we prove $Q_{t-1}(H') \subseteq Q_t(H)$.
Suppose that $t$ is at or before the time that $x$ arrives in $\pi_H$.
Then clearly $Q_{t-1}(H') = Q_{t-1}(H) \subseteq Q_{t}(H)$.  Now
suppose that $t$ is after the time that $x$ arrives in $\pi_H$.  Let
$u$ be the vertex that arrives at time $t-1$ in $\pi_{H'}$.  If $u$ is
not matched by $\ranking(H',\pi_{H'},\sigma_H)$, then $Q_{t-1}(H') =
Q_{t-2}(H') \subseteq Q_{t-1}(H) \subseteq Q_{t}(H)$.  Now suppose
that $u$ is matched by $\ranking(H', \pi_{H'}, \sigma_{H})$, say to
vertex $v'$. We show that $v' \in Q_t(H)$, which by the induction
hypothesis, is enough to prove that $Q_{t-1}(H') \subseteq Q_{t}(H)$.
Note that $u$ arrives at time $t$ in $\pi_H$.  Let $v$ be the vertex
to which $u$ is matched by $\ranking(H, \pi_H, \sigma_H)$.  If $v =
v'$, we are done, so suppose that $v \not= v'$.  Since $v \not\in
Q_{t-1}(H)$, it follows by the induction hypothesis that $v \not\in
Q_{t-2}(H')$.  Therefore, vertex $v$ is available to be matched to $u$
when it arrives in $\pi_{H'}$.  Since $\ranking(H', \pi_{H'},
\sigma_{H})$ matched $u$ to $v'$ instead, $v'$ must have a lower rank
than $v$ in $\sigma_H$.  Since $\ranking(H, \pi_H, \sigma_H)$ chose
$v$, vertex $v'$ must have already been matched when vertex $u$
arrived at time $t$ in $\pi_H$, or, in other words, $v' \in Q_{t-1}(H)
\subseteq Q_t(H)$.

Now suppose that $x \in V$.  In this case, $U_{H} = U_{H'}$ and $\pi_H
= \pi_{H'}$.  Let $R_t(H) \subseteq U_H$ be the set of vertices
matched to vertices in $V$ that are ranked less than or equal to $t$
(under arrival order $\pi_H$ and ranking $\sigma_H$), and let $R_t(H')
\subseteq U_H$ be the set of vertices matched to vertices in $V$ that
are ranked less than or equal to $t$ (under arrival order $\pi_{H}$
and ranking $\sigma_{H'}$).  Then by Lemma~\ref{lem:duality}, we can
apply the same argument as before to show that $R_{t-1}(H') \subseteq
R_t(H)$ for all $t$, which by substituting $t = n$, is sufficient to
prove the claim.
\end{proof}

We define the following notation.  For all $u_G \in U_G$, let
$\zeta^{-1}(u_G)$ be the set of all $u_H \in U_H$ such that
$\zeta(u_H) = u_G$, and for any subset $U_G' \subseteq U_G$, let
$\zeta^{-1}(U_G')$ be the set of all $u_H \in U_H$ such that
$\zeta(u_H) \in U_G'$.  The following lemma shows that we can assume
without loss of generality that $M_G$ is a perfect matching.

\begin{lemma}\label{lem:perfectok}
Let $U' \subseteq U_G$ and $V' \subseteq V$ be the subset of vertices
that are in $M_G$.  Let $G'$ be the subgraph of $G$ induced by $U'
\cup V'$, and let $H'$ be the subgraph of $H$ induced by
$\zeta^{-1}(U') \cup V'$.  Then the expected size of the matching
produced by Ranking on $H'$ is no greater than the expected size of
the matching produced by Ranking on $H$.
\end{lemma}
\begin{proof}
The proof follows by repeated application of
Lemma~\ref{lem:monotonicity} for all $x$ that are not in
$\zeta^{-1}(U') \cup V'$.
\end{proof}

In light of Lemma~\ref{lem:perfectok}, to prove
Theorem~\ref{thm:ranking}, it is sufficient to show that the expected
size of the matching produced by Ranking on $H'$ is at least
$(1-1/e^{1/k} - o(1))|M_G|$.  To simplify notation, we instead assume
without loss of generality that $G = G'$, and hence $G$ has a perfect
matching.  Let $n = OPT_{1P} = |M_G| = |V|$.  Henceforth, fix an arrival order
$\pi$.  To simplify notation, we write $\ranking(\sigma)$ to mean
$\ranking(H, \pi, \sigma)$.

Let $f : U_H \rightarrow V$ be a map such that for all $v \in V$,
there are exactly $k$ vertices $u \in U_H$ such that $f(u) = v$.  The
existence of such a map $f$ follows from the assumption that $G$ has a
perfect matching.  For any vertex $v \in V$ let $f^{-1}(v)$ be the set
of $u \in U_H$ such that $f(u) = v$.  We proceed with the following
two lemmas.

\begin{lemma}\label{lem:easy}
Let $u \in U_H$, and let $v = f(u)$.  For any ranking $\sigma$, if $v$
is not matched by $\ranking(\sigma)$, then $u$ is matched to a vertex
whose rank is less than the rank of $v$ in $\sigma$.
\end{lemma}
\begin{proof}
If $v$ is not matched by $\ranking(\sigma)$, then since there is an
edge between $u$ and $v$, it was available to be matched to $u$ when
it arrived.  Therefore, by the behavior of $\ranking$, $u$ must have
been matched to a vertex of lower rank.
\end{proof}

\begin{lemma}\label{lem:technical}
Let $u \in U_H$, and let $v = f(u)$.  Fix an integer $t$ such that
$1 \leq t \leq n$.  Let $\sigma$ be a permutation, and let $\sigma'$
be the permutation obtained from $\sigma$ by removing vertex $v$ and
putting it back in so its rank is $t$.  If $v$ is not matched by
$\ranking(\sigma')$, then $u$ must be matched by $\ranking(\sigma)$ to
a vertex whose rank in $\sigma$ is less than or equal to $t$.
\end{lemma}
\begin{proof}
For the proof, it is convenient to invoke Lemma~\ref{lem:duality} and
consider $\ranking'(\sigma)$ and $\ranking'(\sigma')$ instead of
$\ranking(\sigma)$ and $\ranking(\sigma')$.  In the process by which
$\ranking'$ constructs its matching, call the moment that the
$t^\textrm{th}$ vertex in $V$ arrives \emph{time} $t$.  For any $1
\leq s \leq n$, let $R_s(\sigma)$ (resp., $R_s(\sigma')$) be the set
of vertices in $U_H$ matched by time $s$ in $\sigma$ (resp.,
$\sigma'$).  By Lemma~\ref{lem:easy}, if $v$ is not matched by
$\ranking(\sigma')$, then $u$ must be matched to a vertex $v'$ in
$\ranking(\sigma')$ such that $\sigma'(v') < \sigma'(v)$.  Hence $u
\in R_{t-1}(\sigma')$.  We prove the lemma by showing that
$R_{t-1}(\sigma') \subseteq R_t(\sigma)$.
Let $\tilde{t}$ be the time that $v$ arrives in $\sigma$.  Then
if $\tilde{t} \geq t$, the two orders $\sigma$ and $\sigma$ are
identical through time $t$, which implies that
$R_{t-1}(\sigma') = R_{t-1}(\sigma) \subseteq R_t(\sigma)$.

Now, in the case that $\tilde{t} < t$, we prove that for $1 \leq s
\leq t$, $R_{s-1}(\sigma') \subseteq R_s(\sigma)$.  The proof, which
is similar to the proof of Lemma~\ref{lem:monotonicity}, proceeds by
induction on $s$.  When $s = 0$, the claim clearly holds, since
$R_0(\sigma') = \emptyset$.  Now, supposing that $R_{s-2}(\sigma')
\subseteq R_{s-1}(\sigma)$, we prove that $R_{s-1}(\sigma') \subseteq
R_s(\sigma)$.  If $s \leq \tilde{t}$, then the two orders $\sigma$ and
$\sigma'$ are identical through time $s$, so $R_{s-1}(\sigma') =
R_{s-1}(\sigma) \subseteq R_s(\sigma)$.  Now suppose that $s >
\tilde{t}$.  Then the vertex that arrives at time $s-1$ in $\sigma'$
is the same as the vertex that arrives at time $s$ in $\sigma$.  Call
this vertex $w$.  If $w$ is not matched by $\ranking'(\sigma')$, then
$R_{s-1}(\sigma') = R_{s-2}(\sigma')$, and we are done by the
induction hypothesis.  Now suppose that $w$ is matched to vertex $x'$
by $\ranking'(\sigma')$ and to vertex $x$ by $\ranking'(\sigma)$.  If
$x = x'$, then again we are done by the induction hypothesis, so
suppose that $x \ne x'$.  Since $x$ was available at time $s-1$ in
$\sigma$, we have $x \not\in R_{s-1}(\sigma)$, and by the induction
hypothesis $x \not\in R_{s-2}(\sigma')$.  Hence, $x$ was available at
time $s-1$ in $\sigma'$.  Since $\ranking'(\sigma')$ matched $w$ to
$x'$, it must be that $\pi(x') < \pi(x)$.  This implies that $x'$ must
be matched when $w$ arrives at time $s$ in $\sigma$, or in other
words, $x' \in R_{s-1}(\sigma) \subseteq R_s(\sigma)$.  By the
induction hypothesis, we are done.
\end{proof}

\begin{lemma}\label{lem:rankingmain}
For $1 \leq t \leq n$, let $x_t$ denote the probability over $\sigma$
that the vertex ranked $t$ in $V$ is matched by $\ranking(\sigma)$.
Then
\begin{equation}\label{eqn:rankingmain}
1 - x_t \leq \frac{1}{kn} \sum_{s = 1}^t x_s \enspace .
\end{equation}
\end{lemma}
\begin{proof}
Let $\sigma$ be permutation chosen uniformly at random, and let
$\sigma'$ be a permutation obtained from $\sigma$ by choosing a vertex
$v \in V$ uniformly at random, taking it out of $\sigma$, and putting
it back so that its rank is $t$.  Note that both $\sigma$ and
$\sigma'$ are distributed uniformly at random among all permutations.
Let $u$ be a vertex chosen uniformly at random from $f^{-1}(v)$.  Note
that conditioned on $\sigma$, $u$ is equally likely to be any of the
$kn$ vertices in $U_H$.  Let $R_t$ be the set of vertices in $U_H$
that are matched by $\ranking(\sigma)$ to a vertex of rank $t$ or
lower in $\sigma$.  Lemma~\ref{lem:technical} states that if $v$ is
not matched by $\ranking(\sigma')$, then $u \in R_t$.  The expected
size of $R_t$ is $\sum_{1 \leq s \leq t} x_s$.  Hence, the probability
that $u \in R_t$, conditioned on $\sigma$, is $(1/(kn)) \sum_{1 \leq s
  \leq t} x_s$.  The lemma follows because the probability that $v$ is
not matched by $\ranking(\sigma')$ is $1 - x_t$.
\end{proof}

We are now ready to prove Theorem~\ref{thm:ranking}.
\begin{proof}[Proof of Theorem~\ref{thm:ranking}.]
For $0 \leq t \leq n$, let $S_t = \sum_{1 \leq s \leq t} x_s$.  Then
the expected size of the matching returned by Ranking on $H$ is $S_n$.
Rearranging (\ref{eqn:rankingmain}) yields, for $1 \leq t \leq n$,
\begin{equation*}
S_t \geq \paren{\frac{kn}{kn+1}}\paren{1 + S_{t-1}},
\end{equation*}
which by induction implies that $S_t \geq \sum_{1 \leq s \leq t}
(kn/(kn+1))^s$, and hence
\begin{equation*}
S_n \geq \sum_{s = 1}^n \paren{\frac{kn}{kn+1}}^s
= kn \paren{1 - \paren{1 - \frac{1}{kn+1}}^n}
= kn \paren{1 - \frac{1}{e^{1/k}} + o(1)} \enspace .
\end{equation*}

\end{proof}

\section{Other Missing Proofs}

In this appendix, we provide the missing proofs from the body of the
paper.

\subsection{Proof of Theorem~\ref{thm:hardness_2paa}}\label{appendix:hardness_2paa}

Fix a constant $c' > c$, and let $n_0$ be the smallest integer such
that for all $n \geq n_0$,
\begin{equation}\label{eqn:asymptotic1}
c' \cdot \frac{c(n^5 + n + 2)}{cn^2 + n + 2} \geq c(n^3 + cn^2 + n + 2)
\end{equation}
and 
\begin{equation}\label{eqn:asymptotic2}
\frac{n/2 + 1}{2} \geq c \enspace .
\end{equation}
Note that since $n_0$ depends only on $c'$, it is a constant.  

We reduce from PARTITION, in which the input is a set of $n \geq n_0$
items, and the weight of the $i$-th item is given by $w_i$.  If $W =
\sum_{i = 1}^n w_i$, then the question is whether there is a partition
of the items into two subsets of size $n/2$ such that the sum of the
$w_i$'s in each subset is $W/2$.  It is known that this problem (even
when the subsets must both have size $n/2$) is NP-hard~\cite{Garey79}.

Suppose that there is an $m/c'$-approximation algorithm to 2PAA($c$);
we will show that constructing the following instance of 2PAA($c$)
(illustrated in Figure~\ref{fig:hardness_2PMBA_small}) allows us to
use the $m/c'$-approximation to solve the PARTITION instance:
\begin{itemize}
\item
First, create $n + 2$ keywords  $c_1, \ldots, c_n, e_1, e_2$.  Second,
create an additional set 
\begin{equation*}
G=\setnot{g_{i,k}}{1 \leq i \leq n^2 \mbox{ and } 1 \leq k \leq c}
\end{equation*}
of $cn^2$ keywords.  The keywords arrive in the order
\begin{equation*}
c_1, \ldots, c_n, e_1, e_2, g_{1,1}, \ldots, g_{1,c}, \ldots \ldots, g_{n^2,1}, \ldots, g_{n^2,c} \enspace .
\end{equation*}

\item
Create $n^2 + 4$ bidders $a, d_1, d_2, f, h_1, \ldots, h_{n^2}$.  Set
the budgets of $a$, $d_1$, and $d_2$ to $cW(1 + n/2)$.  Set the budget
of $f$ to $cW(n^3 + 1)$.  For $1 \leq i \leq n^2$, set the budget of
$h_i$ to $cWn^3$.

\item
For $1 \leq i \leq n$, bidders $a$, $d_1$, and $d_2$ bid $c(w_i + W)$
on keyword $c_i$.

\item
For $j \in \braces{1,2}$, bidder $d_j$ bids $cW$ on keyword $e_j$.
Bidder $f$ bids $cW/2$ on both $e_1$ and $e_2$.

\item
For $1 \leq i \leq n^2$ and $1 \leq k \leq c$, keyword $g_{i,k}$
receives a bid of $W(n^3 + 1)$ from bidder $f$ and a bid of $Wn^3$
from bidder $h_i$.
\end{itemize}
This reduction can clearly be performed in polynomial time.
Furthermore, it can easily be checked that (\ref{eqn:asymptotic2})
implies that no bidder bids more than $1/c$ of its budget on any
keyword.

\begin{figure}
  \begin{center}
    \scalebox{0.9}{
      \includegraphics{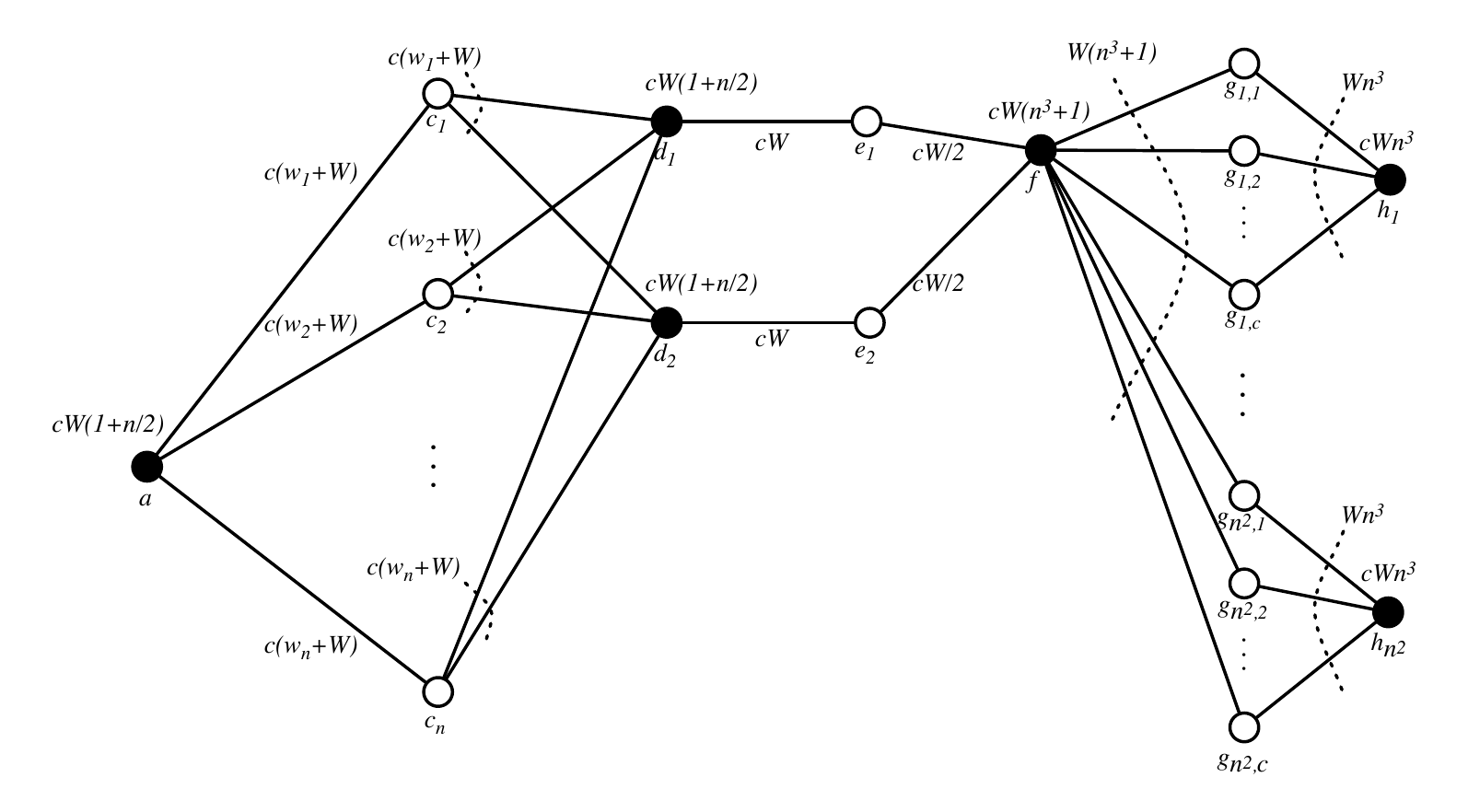}
    }
  \end{center}
  \caption{\small The 2PAA($c$) instance of the reduction.  Each bidder's
    budget is shown above its node, and the bids of bidders for
    keywords is shown near the corresponding edge.  }
  \label{fig:hardness_2PMBA_small}
\end{figure}

We first show that if the PARTITION instance is a ``yes'' instance,
then there exists a feasible solution to the 2PAA($c$) instance of
value at least $cW(n^5 + n + 2)$.  Let $U \subseteq \bracks{n}$ be
such that $|U| = n/2$ and $\sum_{i \in U} w_i = \sum_{i \in
  \overline{U}} w_i = W/2$.

We construct a solution to the 2PAA($c$) instance as follows.  For
every $i \in U$, allocate $c_i$ to $d_1$, and for every $i \in
\overline{U}$, allocate $c_i$ to $d_2$.  For each of these
allocations, choose $a$ as the second-price bidder.  This will reduce
the budget of $d_1$ and $d_2$ to exactly $cW/2$, and hence the bids
from $d_1$ to $e_1$ and from $d_2$ to $e_2$ will be both be reduced to
$cW/2$.  Allocate $e_1$ to $f$ choosing $d_1$ as the second-price
bidder, and allocate $e_2$ to $f$ choosing $d_2$ as the second-price
bidder.  This will reduce the budget of $f$ to $cWn^3$.  The profit
from the solution constructed so far is $cW(n+2)$.  Now allocate
$g_{1,1}, g_{1,2},\ldots,g_{1,c-1}$ to $f$, choosing $h_1$ as the
second-price bidder.  This will reduce the budget of $f$ to $Wn^3$.
Hence, it can act as the second-price bidder for each of the remaining
keywords in $G$.  Allocate $g_{1,c}$ to $h_1$, choosing $f$ as the
second-price bidder, and then, for $2 \leq i \leq n^2$ and $1 \leq k
\leq c$, allocate $g_{i,k}$ to $h_i$, choosing $f$ as the second-price
bidder.  The profit obtained for each keyword in $G$ in this
assignment is $Wn^3$.  Since $|G| = cn^2$, the total profit of the
solution constructed is $cW(n+2) + cWn^5 = cW(n^5 + n + 2)$.

We now show that if there is a second-price matching in the 2PAA($c$)
instance of value at least $cW(n^3 + cn^2 + n + 2)$, then there must
be a partition of $w_1, \ldots, w_n$.  In such a matching, at most
$cW(n + 2)$ units of profit can be obtained from keywords $c_1,
\ldots, c_n, e_1, e_2$, since the initial second-highest bids on those
keywords sum to $cW(n+2)$.  Hence, at least $cW(n^3 + cn^2)$ profit
must come from the keywords in $G$.

Suppose that the budget of $f$ is greater than $cWn^3$ after keywords
$e_1$ and $e_2$ are allocated.  Note that at least $c$ of the keywords
in $G$ must be allocated to reach a profit of $cW(n^3 + cn^2)$ on these
keywords.  Consider what happens after the first $c$ of the keywords
in $G$ are assigned.  For each of these keywords, $f$ must have been
the first-price bidder, so its budget is reduced to an amount greater
than $0$ and less than or equal to $cW$.  Hence, for each keyword in
$G$ allocated henceforth, $f$ is the second-price bidder, and the
profit is at most $cW$.  Since there are at most $c(n^2 - 1)$ more
keywords in $G$, the total profit from the keywords in $G$ is at most
$cWn^3 + c^2W(n^2 - 1)$, which contradicts the fact that at least
$cW(n^3 + cn^2)$ units of profit must come from $G$.  Hence, we
conclude that the budget of $f$ is less than or equal to $cWn^3$ after
keywords $e_1$ and $e_2$ are allocated.

The budget of $f$ can only be smaller than $cWn^3$ if $f$ acts as the
first-price bidder for both $e_1$ and $e_2$.  But this can happen only
if the budgets of both $d_1$ and $d_2$ are reduced to an amount less
than or equal to $cW/2$.  For $j \in \braces{1,2}$, let $U_j \subseteq
\bracks{n}$ be the set of indices $i$ such that $d_j$ acts as the
first-price bidder for $i$.  For both $j$, we have that
\begin{equation}\label{eqn:impliespartition}
\sum_{i \in U_j} c(w_i + W) \geq \frac{cW}{2} + \frac{cWn}{2} \enspace .
\end{equation}
Rearranging (\ref{eqn:impliespartition}) yields $\sum_{i \in U_j} W
\geq W/2 + Wn/2 - \sum_{i \in U_j} w_i$, which implies $W |U_j| \geq
W/2 + Wn/2 - W$, and hence $|U_j| \geq n/2 - 1/2$.  By integrality,
then, $|U_j| \geq n/2$ for both $j$.  Hence $|U_j| = n/2$ for both
$j$, and using (\ref{eqn:impliespartition}) again, we have
$\sum_{i \in U_j} c w_i + cW|U_j| \geq cW/2 + cWn/2$
which implies that $\sum_{i \in U_j} w_i \geq W / 2$ for both $j$.
Therefore, the partition defined by $U_1$ and $U_2$ is a solution to
the PARTITION instance.

To conclude the proof, note that the number of keywords in the
2PAA($c$) instance is $cn^2 + n + 2$.  Hence, if the PARTITION
instance is a ``yes'' instance, then by (\ref{eqn:asymptotic1}), we
can run the $m/c'$-approximation algorithm to find a second-price
matching of value at least $cW(n^3+cn^2+n+2)$, and if the
PARTITION instance is a ``no'' instance, then the value of the
solution returned by the algorithm must be strictly less than $cW(n^3
+ cn^2 + n+ 2)$.  Hence, an $m/c'$-approximation algorithm for 2PAA($c$)
can be used to solve PARTITION.

\subsection{Proof of Theorem~\ref{thm:hardness_2pm}}\label{appendix:hardness_2pm}

  The proof is by a reduction from vertex cover.  Given a graph $G$,
  we construct an instance $f(G)$ of 2PM as follows. First, for each
  edge $e \in E(G)$, we create a keyword with the same label (called
  an \emph{edge keyword}), and for each vertex $v \in V(G)$, we create
  a bidder with the same label (called a \emph{vertex bidder}). Bidder
  $v$ bids for keyword $e$ if vertex $v$ is one of the two end points
  of edge $e$.  (Recall that in 2PM, if a bidder makes a non-zero bid
  for a keyword, that bid is $1$.)  In addition, for each edge $e$, we
  create a unique bidder $x_e$ who also bids for $e$. Furthermore, for
  each vertex $v$, we create a gadget containing two keywords $h_v$
  and $l_v$ and two bidders $y_v$ and $z_v$. We let $v$ and $y_v$ bid
  for $h_v$; and $y_v$ and $z_v$ bid for $l_v$.  The keywords arrive
  in an order such that for each $v \in V(G)$, keyword $h_v$ comes
  before $l_v$, and the edge keywords arrive after all of the $h_v$'s
  and $l_v$'s have arrived.  An example of this reduction is shown in
  Figure~\ref{fig:hardness_2pm}.

  \begin{figure}
    \begin{center}
      \begin{tabular}{cc}
        \subfigure[]{\scalebox{0.7}{\label{fig:hardness_2pm:a}
            \includegraphics{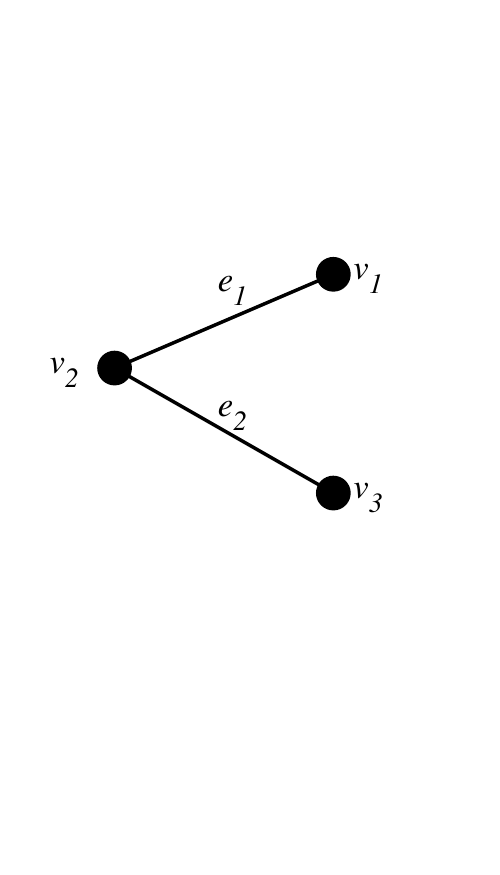}
          } 
        } &
        \subfigure[]{\scalebox{0.7}{\label{fig:hardness_2pm:b}
            \includegraphics{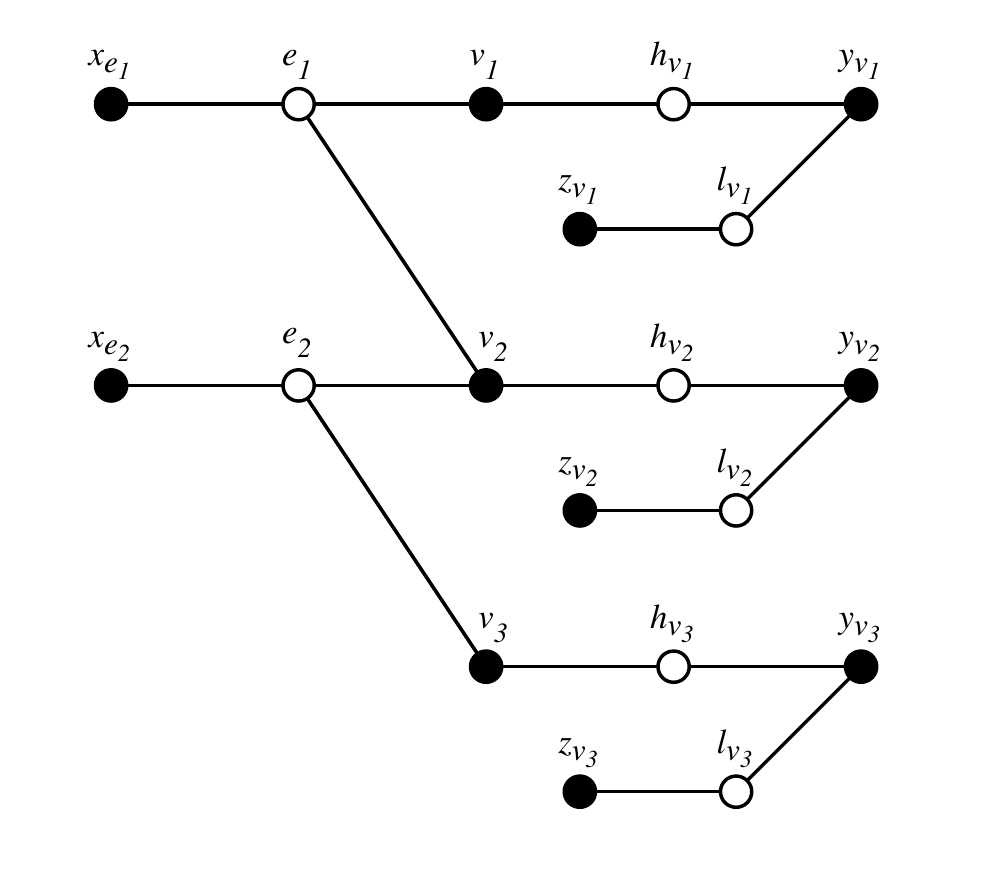}
          }
        }
      \end{tabular}
    \end{center}
    \caption{\small The reduction from an instance $G$ of vertex cover
      (Figure~\ref{fig:hardness_2pm:a}) to an instance $f(G)$ of
      2PM (Figure~\ref{fig:hardness_2pm:b}).
    } \label{fig:hardness_2pm}
  \end{figure}

  The following lemma provides the basis of the proof.
  \begin{lemma}\label{lem:hardness_2pm}
    Let $OPT_{VC}$ and $OPT_{2P}$ be the size of the minimum vertex
    cover of $G$ and the maximum second-price matching on $f(G)$,
    respectively. Then
    $$
      OPT_{2P} = 2|V(G)| + |E(G)| - OPT_{VC}
    $$
  \end{lemma}
\begin{proof}
  We first show that given a vertex cover $S$ of size $OPT_{VC}$ of
  $G$, we can construct a solution to the 2PM instance whose value is
  $2|V(G)| + |E(G)| - OPT_{VC}$. For each vertex $v \notin S$, we
  allocate $h_v$ to $v$ (with $y_v$ acting as the second-price bidder)
  and $l_v$ to $y_v$ (with $z_v$ acting as the second-price bidder),
  getting a profit of $2$ from the gadget for $v$.  For each vertex $v
  \in S$, we allocate $h_v$ to $y_v$ (with $v$ acting as the
  second-price bidder) and ignore $l_v$, getting a profit of $1$ from
  the gadget for $v$.  We then allocate each edge keyword $e$ to
  $x_e$.  During each of these edge keyword allocations, at least one
  of the two vertex bidders that bid for $e$ is still available, since
  $S$ is a vertex cover.  Hence, for each of these allocations, there
  is a bidder that can act as a second-price bidder, and the profit
  from the allocation is $1$.  This allocation yields a second-price
  matching of size $2|V(G)| + |E(G)| - OPT_{VC}$.  Therefore,
  $OPT_{2P} \geq 2|V(G)| + |E(G)| - OPT_{VC}$.
  
  To show that $OPT_{2P} \leq 2|V(G)| + |E(G)| - OPT_{VC}$, we start
  with an optimal solution to $f(G)$ of value $OPT_{2P}$ and
  construct a vertex cover of $G$ of size $2|V(G)| + |E(G)| -
  OPT_{2P}$.  To do this, we first claim that there exists an optimal
  solution of $f(G)$ in which every edge-keyword is allocated for a
  profit of $1$. Consider any optimal second-price matching of the
  instance.  Let $e$ be an edge-keyword $e$ that is not allocated for
  a profit of $1$.  If it is adjacent to a vertex bidder that is
  unassigned when $e$ arrives, then $e$ can be allocated to $x_e$ for
  a profit of $1$, which can only increase the value of the solution.
  Suppose, on the other hand, that both of its vertex bidders are not
  available when $e$ arrives.  Let $v$ be a vertex bidder that bids
  for $e$.  Since it is not available, $h_v$ must have been allocated
  to $v$.  We can transform this second-price matching to another one in
  which $h_v$ is assigned to $y_v$, $l_v$ is ignored and $e$ is
  assigned to $x_e$, with $v$ acting the second-price bidder in both
  cases.  This does not decrease the total profit of the solution.
  Hence, we can perform these transformations for each edge keyword
  $e$ that is not allocated for a profit of $1$ until we obtain a new
  optimal solution in which each edge keyword is allocated for a
  profit of $1$.

  Now consider an optimal second-price matching in which all edge
  keywords are allocated for a profit of $1$. Let $T \subseteq V(G)$
  be the set of vertices represented by vertex bidders that are not
  allocated any keywords in this second-price matching.  Then $|T| =
  2|V(G)| + |E(G)| - |OPT_{2P}|$, and $T$ is a vertex cover, which
  implies $OPT_{VC} \leq 2|V(G)| + |E(G)| - OPT_{2P}$. The lemma
  follows.
\end{proof}

Now, suppose that we have an $\alpha$-approximation for 2PM.  We will
show how to use this approximation algorithm and our reduction to
obtain an $((8\alpha - 7)/\alpha)$-approximation for Vertex Cover on
4-regular graphs.  By Theorem~\ref{thm:hardness_vc}, this means that
$(8\alpha - 7)/\alpha \geq 53/52$, and hence $\alpha \geq 364/363$,
unless $P = NP$.

To construct this $((8\alpha - 7)/\alpha)$-approximation algorithm,
given a 4-regular graph $G$, run the above reduction to obtain a 2PM
instance $f(G)$.  Then use the $\alpha$-approximation to obtain a
second-price matching $M$ whose value is at least $OPT_{2P}/\alpha$.
Now, just as in the proof of Lemma~\ref{lem:hardness_2pm}, we can
assume that in $M$, every edge keyword $e$ is allocated to $x_e$.
Hence, the set of vertices $T$ associated with the vertex bidders that
are not allocated a keyword form a vertex cover, and
\begin{eqnarray}
|T| 
& \leq & 2|V(G)| + |E(G)| - OPT_{2P}/\alpha \nonumber \\
& =    & 2|V(G)| + |E(G)| - (2|V(G)| + |E(G)| - OPT_{VC})/\alpha  \nonumber \\
& =    & (1-1/\alpha)(2|V(G)| + |E(G)|) + OPT_{VC}/\alpha \label{eqn:h2pm}
\end{eqnarray}
Since $G$ is $4$-regular, we have $OPT_{VC} \geq m/4 =
(2|V(G)|+|E(G)|)/8$, and hence by (\ref{eqn:h2pm}), we conclude that
$|T| \leq ((8\alpha - 7)/\alpha) OPT_{VC}$, which finishes the proof
of the theorem.

\subsection{Proof of Theorem~\ref{thm:2_approx}}\label{appendix:2_approx}

Since the number of vertices matched by $f$ is an upper bound on the
profit of the maximum second-price matching on $G$, we need only to
prove that the second-price matching contains at least half of the
keywords matched by $f$.  By the behavior of the algorithm, it is
clear that whenever a vertex $u$ is matched to $f(u)$ in the
second-price matching, the profit obtained is $1$.  Furthermore, every
time an an edge is removed from $f$, a new keyword is added to the
second-price matching. Thus, the theorem follows.

\subsection{Proof of Theorem~\ref{thm:rand_lower_online2pm}\label{appendix:rand_lower_online2pm}}

  We invoke Yao's Principle~\cite{Yao77} and construct a distribution
  of inputs for which the best deterministic algorithm achieves an
  expected performance of (asymptotically) $1/2$ the value of the
  optimal solution.

  Our distribution is constructed as follows.  The first keyword
  arrives, and it is adjacent to two bidders.  Then the second keyword
  arrives, and it is adjacent to one of the two bidders adjacent to
  the first keyword, chosen uniformly at random, as well as a new
  bidder; then the third keyword arrives, and it is adjacent to one of
  the bidders adjacent to the second keyword, chosen uniformly at
  random, as well as a new bidder; and so on, until the $m$-th keyword
  arrives.  We call this a \emph{normal} instance.  To analyze the
  performance of the online algorithms, we also define a
  \emph{restricted} instance to be one that is exactly the same as a
  normal instance except that one of the two bidders of the first
  keyword is marked \emph{unavailable}, i.e., he can not participate
  in any auction.
  
  Clearly, an offline algorithm that knows the random choices beforehand
  can allocate each keyword to the bidder that will not be adjacent to 
  the keyword that arrives next.  In this way, it can ensure that for each
  keyword, there is a bidder that can act as a second-price bidder.
  Hence for a normal instance, the optimal second-price matching obtains
  a profit of $m$.

  Consider the algorithm Greedy, which allocates a keyword to an
  arbitrary adjacent bidder if and only if there is another available
  bidder to act as a second-price bidder.  Our proof consists of two
  steps: first, we will show that the expected performance of Greedy
  on the normal instance is $(m + 1)/2$, and second we will prove that
  Greedy is the best algorithm in expectation for both types of
  instances.

  Let $X_k^*$ and $Y_k^*$ be the \emph{expected} profit of Greedy on a
  normal and a restricted instance of $k$ keywords, respectively
  (where $X_0^*$ and $Y_0^*$ are both defined to be $0$). Given the
  first keyword of a normal instance, Greedy allocates it to an
  arbitrary bidder. Then, with probability $1/2$, it is faced with a
  normal instance of $k-1$ keywords, and with probability $1/2$, it is
  faced with a restricted instance of $k-1$ keywords. Therefore, for
  all integers $k \geq 1$,
  \begin{equation}
    \label{eq:1}
    X_{k}^* = 1/2(X_{k-1}^* + Y_{k-1}^*) + 1 \enspace .
  \end{equation}
  
  On the other hand, given the first keyword of a restricted instance,
  Greedy just waits for the second keyword.  Then, with probability
  $1/2$, the second keyword chooses the marked bidder, giving Greedy a
  restricted instance of $k-1$ keywords, and with probability $1/2$,
  the second keyword chooses the unmarked bidder, giving Greedy a
  normal instance of $k-1$ keywords. Therefore, for all $k$,
  \begin{equation}
    \label{eq:2}
    Y_k^* = 1/2(X^*_{k-1} + Y^*_{k-1}) \enspace .
  \end{equation}
  
  From (\ref{eq:1}) and (\ref{eq:2}) we have, for all $k$,
  \begin{equation}
    \label{eq:3}
    Y_k^* = X_k^* - 1 \enspace .
  \end{equation}
  Plugging 
  (\ref{eq:3}) for $k = m-1$ into (\ref{eq:1}) for $k = m$ yields
  \begin{equation}
    \label{eq:4}
    X_m^* = X_{m-1}^* + 1/2 \enspace , 
  \end{equation}
  and hence, by induction $X_m^* = (m+1)/2$.

  Now, we prove that Greedy is the best among all algorithms on these
  two types of instances. In fact, we make it easier for the
  algorithms by telling them beforehand how many keywords in the
  instance they will need to solve. Let $X_m$ and $Y_m$ be the
  expected number of keywords in the second-price matching produced by
  the \emph{best} algorithms that ``know'' that they are solving a
  normal instance of size $m$ and a restricted instance of size $m$,
  respectively. Let ${\cal A}_m$ and ${\cal B}_m$ denote these optimal
  algorithms.

  We prove that $X_m \leq X_m^*$ and $Y_m \leq Y_m^*$ for all $m$ by
  induction. The base case in which $m=1$ is easy, since no algorithm
  can obtain a profit of more than one on a normal instance of one
  keyword or more than zero on a restricted instance of one keyword.
  We now prove the induction step.

  First, consider ${\cal A}_m$. When the first keyword arrives, ${\cal
    A}_m$ has two choices: either ignore it or allocate it to one of the
  bidders. If ${\cal A}_m$ ignores the first keyword, its performance
  is at most the performance of ${\cal A}_{m-1}$ on the remaining
  keywords, which constitute a normal instance of $m-1$ keywords. On
  the other hand, if ${\cal A}_m$ allocates the first keyword to one of
  the bidders, then with probability $1/2$, it is faced with a normal
  instance of $m-1$ keywords, and with probability $1/2$ it is faced
  with a restricted instance of $m-1$ keywords.  The performance of
  ${\cal A}_m$ on these instance is at most the performance of ${\cal
    A}_{m-1}$ and ${\cal B}_{m-1}$, respectively. Thus, by the induction
  hypothesis, (\ref{eq:3}), and (\ref{eq:4}), we have
  \begin{eqnarray*}
    X_m & \leq & \max\{X_{m-1}, 1/2(X_{m-1} + Y_{m-1}) + 1\} \\
    & \leq & \max\{X_{m-1}^*, 1/2(X_{m-1}^* + Y_{m-1}^*) + 1\} \\
    & =    & \max\{X_{m-1}^*, 1/2(X_{m-1}^* + X_{m-1}^* - 1) + 1\} \\
    & =    & X_{m-1}^* + 1/2 \\
    & =    & X_m^* \enspace .
  \end{eqnarray*}
  
  Next, consider ${\cal B}_m$. When the first keyword arrives, ${\cal
    B}_m$ cannot allocate it for a profit.  If it allocates it for a
  profit of $0$, then it is faced with a restricted instance of $m-1$
  keywords.  If it does not allocate the keyword, then with
  probability 1/2, ${\cal B}_m$ is faced with a normal instance of
  $m-1$ keywords, and with probability $1/2$, it is faced with a
  restricted instance of $m-1$ keywords. Its performance on these
  instances is at most those of ${\cal A}_{m-1}$ and ${\cal B}_{m-1}$,
  respectively. Thus, by the induction hypothesis and (\ref{eq:2}), we
  have
  \begin{eqnarray*}
    Y_m & \leq & \max\{Y_{m-1},1/2(X_{m-1} + Y_{m-1})\} \\
    & \leq & \max \{Y_{m-1}^*,1/2(X_{m-1}^* + Y_{m-1}^*)\} \\
    & = & Y_m^* \enspace .
  \end{eqnarray*}
  This completes our proof.

\end{document}